\documentclass[10pt]{article}
\usepackage{amsmath}
\usepackage{amssymb, amscd, amsthm}
\usepackage[all]{xy}
\usepackage[dvips]{graphicx}
\usepackage{verbatim}
\usepackage[perpage,symbol*]{footmisc}
\newtheorem{theorem}{Theorem}

\newtheorem{lemma}{\textbf{Lemma}}[section]
\newtheorem{remark}{\textbf{Remark}}[section]
\newtheorem{corollary}{\textbf{Corollary}}[section]

\newtheorem{example}{\textbf{Example}}[section]

\usepackage[justification=centering]{caption}
\usepackage{longtable}
\newcommand{\tabincell}[2]{\begin{tabular}{@{}#1@{}}#2\end{tabular}}
\usepackage{multirow}
\usepackage{slashbox}

\newcommand{\F}{\mathbb{F}}

\begin{document}

\baselineskip 17pt
\title{\Large\bf New MDS Euclidean Self-orthogonal Codes}

\author{\large  Xiaolei Fang \quad\quad Meiqing Liu \quad\quad Jinquan Luo*}\footnotetext{The authors are with School of Mathematics
and Statistics \& Hubei Key Laboratory of Mathematical Sciences, Central China Normal University, Wuhan China 430079.\\
 E-mail: fangxiaolei@mails.ccnu.edu.cn(X.Fang), 15732155720@163.com(M.Liu), luojinquan@mail.ccnu.edu.cn(J.Luo)}

\date{}
\maketitle

{\bf Abstract}:
   In this paper, a criterion of MDS Euclidean self-orthogonal codes is presented. New MDS Euclidean self-dual codes and self-orthogonal codes are constructed via this criterion. In particular, among our constructions, for large square $q$, about $\frac{1}{8}\cdot q$ new MDS Euclidean (almost) self-dual codes over $\F_q$ can be produced. Moreover, we can construct about $\frac{1}{4}\cdot q$ new MDS Euclidean self-orthogonal codes with different even lengths $n$ with dimension $\frac{n}{2}-1$.

\section{Introduction}

 \quad\; Let $q$ be a prime power and $\mathbb{F}_q$ be a finite field of cardinality $q$. A linear code $C$ over $\F_q$, denote by $[n,k,d]_q$, is a
linear $\F_q$-subspace of $\F_q^n$ with dimension $k$ and minimal (Hamming) distance $d$. If the parameters can reach the Singleton bound, that is
$n+1=k+d$, then we call $C$ a maximum distance separable(MDS) code. Denote the Euclidean dual code of $C$ by $C^{\perp}$. If $C\subseteq C^\perp$,
we call $C$ a Euclidean self-orthogonal code. If $C=C^\perp$, then $C$ is called a Euclidean self-dual code.

Both MDS codes and Euclidean self-dual codes have theoretical and practical significance. The study on MDS Euclidean self-dual codes have attracted a
lot of interest, especially for the constructions of MDS Euclidean self-dual codes. Such codes can be constructed in various ways, which
mainly are: (1). orthogonal designs, see [\ref{GK}, \ref{HK1}, \ref{HK2}]; (2). building up technique, see [\ref{KL1}, \ref{KL2}]; (3). constacyclic
codes, see [\ref{KZT}, \ref{TW}, \ref{YC}]; (4). (extended) GRS codes, see [\ref{FF3}, \ref{GKL}, \ref{JX2}, \ref{LLL}, \ref{TW}, \ref{Yan}, \ref{ZF}].

It is noted that the parameters of MDS Euclidean self-dual codes can be completely determined by the code length $n$. For the finite field of even
characteristic, Grassl and Gulliver completely determined the MDS Euclidean self-dual codes of length $n$ for all $1\leq n\leq q$ in [\ref{GG}]. In
[\ref{GUE}, \ref{KZT}, \ref{TW}], the authors obtained some new MDS Euclidean self-dual codes through cyclic, constantcyclic and negacyclic codes.
In [\ref{JX2}], Jin and Xing firstly presented a general and efficient method to construct MDS Euclidean self-dual codes by utilizing GRS codes.
Afterwards, many researchers started to construct MDS Euclidean self-dual codes via GRS codes. In [\ref{ZF}], Zhang and Feng stated that when $q\equiv3\,(\mathrm{mod}\,4)$ and $n\equiv2\,(\mathrm{mod}\,4)$, MDS Euclidean self-dual code of length $n$ does not exist. We list all the known results on the systematic constructions of MDS Euclidean self-dual codes, which are depicted in Table 1.

\begin{center}
\begin{longtable}{|c|c|c|}  
\caption{Known systematic construction on MDS Euclidean self-dual codes of length
$n$\\ ($\eta$ is the quadratic character of $\mathbb{F}_{q}$) } \\ \hline
$q$ & $n$ even & Reference\\  \hline
$q$ even  &  $n \leq q$   & [\ref{GG}] \\ \hline
$q$ odd & $n=q+1$ & [\ref{GG}]\\ \hline $q$ odd & $(n-1)|(q -1)$, $\eta(1 - n) = 1$ &   [\ref{Yan}] \\ \hline
$q$ odd & $(n-2)|(q - 1)$, $\eta(2 - n) = 1$ &   [\ref{Yan}]\\ \hline
$q = r^{s}\equiv3\,(\mathrm{mod}\,4)$ &  $n-1= p^m \mid(q-1)$, prime $p\equiv3\,(\mathrm{mod}\,4)$ and  $m$ odd &  [\ref{GUE}]\\ \hline
$q = r^{s}$, $r\equiv1\,(\mathrm{mod}\,4)$, $s$ odd &  $n-1= p^m\mid (q-1)$, $m$ odd and prime $p\equiv1\,(\mathrm{mod}\,4)$  &  [\ref{GUE}]\\ \hline
$q = r^{s}$ , $r$ odd, $s\geq 2$ & $n = lr$,  $l$ even and $2l|(r - 1)$ &   [\ref{Yan}] \\ \hline

$q = r^{s}$ , $r$ odd, $s \geq 2$ & $n = lr$,  $l$ even , $(l - 1)|(r - 1)$ and $\eta(1 - l)=1$ &   [\ref{Yan}] \\ \hline

$q = r^{s}$ , $r$ odd, $s \geq 2$ & $n = lr + 1$, $l$ odd , $l|(r - 1)$ and $\eta(l) = 1$  &   [\ref{Yan}] \\ \hline
 $q = r^{s}$ , $r$ odd, $s \geq 2$ & $n = lr + 1$, $l$ odd , $(l - 1)|(r - 1)$ and $\eta(l - 1) = \eta(-1) = 1$ &  [\ref{Yan}] \\ \hline

$q=r^2$  & $n \leq r$  & [\ref{JX2}] \\ \hline
$q = r^2, r\equiv3\,(\mathrm{mod}\,4)$  &  $n= 2tr$ for any $t \leq \frac{r - 1}{2}$ &   [\ref{JX2}]\\ \hline

$q = r^2$, $r$ odd & $n = tr$, $t$ even and $1 \leq t \leq r$ &   [\ref{Yan}] \\ \hline

 $q = r^2$, $r$ odd & $n = tr + 1$,  $t$ odd and $1 \leq t \leq r$ &   [\ref{Yan}] \\ \hline

$q \equiv1\,(\mathrm{mod}\,4)$ &  $ n|(q - 1), n < q - 1$ &   [\ref{Yan}] \\ \hline
$q\equiv1\,(\mathrm{mod}\,4)$ &  $4^{n}\cdot n^{2} \leq q$ &  [\ref{JX2}]\\ \hline

  $q = p^k $, odd prime $p$ & $n= p^r+1$, $r|k$ &   [\ref{Yan}] \\ \hline
$q = p^k $, odd prime $p$ & $n= 2p^e$, $1 \leq e < k$, $\eta(-1) = 1$&  [\ref{Yan}] \\ \hline
$q=r^2$, $r$ odd & $n=tm$, $1\leq t \leq \frac{r-1}{\gcd(r-1,m)}$, $\frac{q-1}{m}$ even &  [\ref{LLL}] \\ \hline
$q=r^2$, $r$ odd & $n=tm+1$, $tm$ odd, $1\leq t \leq \frac{r-1}{\gcd(r-1,m)}$ and $m|(q-1)$  & [\ref{LLL}]\\ \hline
$q=r^2$, $r$ odd & $n=tm+2$, $tm$ even, $1\leq t \leq \frac{r-1}{\gcd(r-1,m)}$ and $m|(q-1)$   &   [\ref{LLL}]\\\hline
$q=r^2$, $r$ odd & $n=tm$, $1\leq t \leq \frac{r+1}{\gcd(r+1,m)}$, $\frac{q-1}{m}$ even & [\ref{FLLL}] \\ \hline

$q=r^2$, $r$ odd  &\tabincell{c}{$n=tm+2$, $tm$ even(except when $t$ is even, $m$ is even\\
 and $r\equiv1\,(\mathrm{mod}\,4)$), $1\leq t \leq \frac{r+1}{\gcd(r+1,m)}$ and $m|(q-1)$}   &  [\ref{FLLL}]  \\\hline
 $q=r^2$, $r$ odd & $n=tm+1$, $tm$ odd, $2\leq t \leq \frac{r+1}{2\gcd(r+1,m)}$ and $m|(q-1)$  & [\ref{FLLL}] \\ \hline
   $q=r^2$, $r$ odd & \tabincell{c}{$n=tm$, $1\leq t \leq \frac{s(r-1)}{\gcd(s(r-1),m)}$, $s$ even, $s|m$,\\ $\frac{r+1}{s}$ even and $\frac{q-1}{m}$ even}  & [\ref{FLLL}] \\ \hline
      $q=r^2$, $r$ odd & \tabincell{c}{$n=tm+2$, $1\leq t \leq \frac{s(r-1)}{\gcd(s(r-1),m)}$, $s$ even, $s|m$,\\ $s\mid r+1$ and $m|(q-1)$}  & [\ref{FLLL}] \\ \hline
$q=p^{m}$, $m$ even, odd prime $p$ & $n=2tr^l$ with $r=p^s$, $s\mid\frac{m}{2}$, $0\leq l\leq \frac{m}{s}$ and $1\leq t\leq\frac{r-1}{2}$ & [\ref{FF3}]\\ \hline
$q=p^{m}$, $m$ even, odd prime $p$ &\tabincell{c}{$n=(2t+1)r^l+1$ with $r=p^s$, $s\mid\frac{m}{2}$, $0\leq l<\frac{m}{s}$ \\and $0\leq t\leq\frac{r-1}{2}$ or $l=\frac{m}{s}$, $t=0$} & [\ref{FF3}]\\ \hline
$q=p^m\equiv1\,(\mathrm{mod}\,4)$ & $n= p^l+1$ with $0\leq l\leq m$ &   [\ref{FF3}] \\ \hline
\end{longtable}
 \end{center}

For large $q$, the number of different lengths for MDS self-dual codes in Table 1 is much smaller than $\frac{q+1}{2}$.
However, in this paper, we produce
approximately $\frac{q}{8}$ many $q$-ary MDS Euclidean self-dual codes with different lengths. In particular, for odd length, there does not exist self-dual codes. The extreme case is almost self-dual. Self-orthogonality is a nature generalization of self-duality. A criterion of MDS Euclidean self-orthogonal via (extended) GRS codes is presented and applied to construct new MDS Euclidean self-orthogonal codes in this paper.

The rest of this paper is organized as follows. In Section 2, we introduce some basic notations and useful results on self-dual codes, self-orthogonal codes and GRS codes. A criterion of MDS self-orthogonal codes is also presented. In Section 3, we give constructions of new MDS Euclidean self-dual codes utilizing GRS codes. In Section 4, some new MDS self-orthogonal codes are constructed by applying the criterion of MDS self-orthogonal codes. We conclude the results in Section 5.

\section{Preliminaries}
 \quad\; In this section, we introduce some basic notations and useful results on self-dual codes, self-orthogonal codes and GRS codes.

For two vectors $\overrightarrow{b}=(b_1,b_2,\ldots,b_n)$ and $\overrightarrow{c}=(c_1,c_2,\ldots,c_n)$ of $\F_q^n$, the Euclidean inner product is $$\left\langle\overrightarrow{b},\overrightarrow{c}\right\rangle=b_1c_1+b_2c_2+\cdots+b_nc_n.$$
Let $C$ be a linear code over $\F_q$. The Euclidean dual of $C$ is defined as follows:
$$C^\perp=\left\{\overrightarrow{u}\in\F_q^n|\,\text{for\, all}\,\, \overrightarrow{c}\in C, \langle\overrightarrow{u},\overrightarrow{c}\rangle=0\right\}.$$
If $C\subseteq C^\perp$, we call $C$ a Euclidean self-orthogonal code. If $C=C^\perp$, then $C$ is called a Euclidean self-dual code. It is noted that if
$C$ is Euclidean self-dual, the length of $C$ is even. For odd length, when $C\subseteq C^\perp$ and $\dim(C^\perp)=\dim(C)+1$, we call $C$ almost
self-dual.

For $1<n<q$, we choose two $n$-tuples $\overrightarrow{v}=(v_{1},v_{2},\ldots,v_{n})$ and $\overrightarrow{a}=(a_{1},a_{2},\ldots,a_{n})$, where
$v_{i}\in\mathbb{F}_{q}^{*}$, $1\leq i\leq n$ ($v_{i}$ may not be distinct) and $a_{i}$, $1\leq i\leq n$ are distinct elements in $\mathbb{F}_{q}$.
Then the GRS code of length $n$ associated with $\overrightarrow{v}$ and $\overrightarrow{a}$ is defined below:
\begin{equation}\label{def GRS}
\mathbf{GRS}_{k}(\overrightarrow{a},\overrightarrow{v})=\left\{(v_{1}f(a_{1}),\ldots,v_{n}f(a_{n})):f(x)\in\mathbb{F}_{q}[x],\mathrm{deg}(f(x))\leq k-1\right\},
\end{equation}
where $1\leq k\leq n$.

It is well-known that the code $\mathbf{GRS}_{k}(\overrightarrow{a},\overrightarrow{v})$ is a $q$-ary $[n,k]$ MDS code and its dual code is also MDS
[\ref{MS}, Chapter 11].

Moreover, the extended GRS code associated with $\overrightarrow{v}$ and $\overrightarrow{a}$ is defined by:
\begin{equation}\label{def extended GRS}
\mathbf{GRS}_{k}(\overrightarrow{a},\overrightarrow{v},\infty)=\left\{(v_{1}f(a_{1}),\ldots,v_{n}f(a_{n}),f_{k-1}):f(x)\in\mathbb{F}_{q}[x],
\mathrm{deg}(f(x))\leq k-1\right\},
\end{equation}
where $1\leq k\leq n$ and $f_{k-1}$ is the coefficient of $x^{k-1}$ in $f(x)$.

It is also well known that the code $\mathbf{GRS}_{k}(\overrightarrow{a},\overrightarrow{v},\infty)$ is a $q$-ary $[n+1,k]$ MDS code and its dual is also
MDS [\ref{MS}, Chapter 11].

We define
\begin{equation*}\label{def GRS}
L_{\overrightarrow{a}}(a_{i})=\prod_{1\leq j\leq n,j\neq i}(a_{i}-a_{j})
\end{equation*}
and let $QR_{q}$ be the set of nonzero squares of $\mathbb{F}_{q}$.
We give the following lemmas, which are useful in the proof of the main results.

\begin{lemma}([\ref{LXW}], Lemma 5)\label{y3}
Let $a_1,a_2,\cdots,a_n$ be distinct elements of $\mathbb{F}_{q}$. Then we have
\begin{equation*}
\sum_{i=1}^na_i^m\cdot L_{\overrightarrow{a}}(a_{i})^{-1}=
\begin{cases}
0, \text{ $0\leq m\leq n-2$;}\\
1, \text{ $ m=n-1$.}
\end{cases}
\end{equation*}
\end{lemma}

A criterion of MDS Euclidean self-orthogonal codes is presented in the following lemmas, which can be regarded as generalizations of Corollary 2.4
in [\ref{JX2}] and Lemma 2 in [\ref{Yan}], and be applied to construct new MDS Euclidean self-orthogonal codes.
\begin{lemma}\label{GRS}
Assume $1\leq k\leq \lfloor\frac{n}{2}\rfloor$. The code $\mathbf{GRS}_k(\overrightarrow{a},\overrightarrow{v})$ is Euclidean self-orthogonal if and only if there exists a nonzero polynomial $\lambda(x)=\lambda_{n-2k}x^{n-2k}+\cdots+\lambda_1 x+\lambda_0\in \mathbb{F}_{q}[x]$ such that
$v_i^2=\lambda(a_i)L_{\overrightarrow{a}}(a_i)^{-1}\neq0$ for any $1\leq i\leq n$.
\end{lemma}
\begin{proof}
Let $(v_1a_1^i,v_2a_2^i,\ldots,v_na_n^i)$ be a basis of $\mathbf{GRS}_k(\overrightarrow{a},\overrightarrow{v})$ with $0\leq i\leq k-1$.
\begin{equation*}
\begin{aligned}
&\mathbf{GRS}_{k}(\overrightarrow{a},\overrightarrow{v})\,\text{is self-orthogonal}&\\
\Leftrightarrow &v_{1}^{2}a_{1}^{m}+\ldots+v_{n}^{2}a_{n}^{m}=0, \,\,\text{for}\,\, 0\leq m \leq 2k-2&\\
\Leftrightarrow &\left(
\begin{array}{cccc}
1&1&\cdots&1\\
a_{1}&a_{2}&\cdots&a_{n}\\
\vdots&\vdots&\ddots&\vdots\\
a_{1}^{2k-2}&a_{2}^{2k-2}&\cdots&a_{n}^{2k-2}\\
\end{array}
\right)
\left(
\begin{array}{c}
v_{1}^{2}\\
v_{2}^{2}\\
\vdots\\
v_{n}^{2}
\end{array}
\right)
=\overrightarrow{0}.&
\end{aligned}
\end{equation*}
Denote by $v_i^2=x_i$ with $1\leq i\leq n$. By Lemma \ref{y3}, the solution of
\begin{equation*}
\left(
\begin{array}{cccc}
1&1&\cdots&1\\
a_{1}&a_{2}&\cdots&a_{n}\\
\vdots&\vdots&\ddots&\vdots\\
a_{1}^{n-2}&a_{2}^{n-2}&\cdots&a_{n}^{n-2}\\
\end{array}
\right)
\left(
\begin{array}{c}
x_{1}\\
x_{2}\\
\vdots\\
x_{n}\\
\end{array}
\right)
=\overrightarrow{0}
\end{equation*}
is
\begin{equation*}
\lambda \cdot \left(\frac{1}{L_{\overrightarrow{a}}(a_1)},\frac{1}{L_{\overrightarrow{a}}(a_2)},\ldots,\frac{1}{L_{\overrightarrow{a}}(a_n)}\right)^T,
\end{equation*}
where $\lambda\in\F_q$. Let
\begin{center}
$A=\left(
\begin{array}{cccc}
1&1&\cdots&1\\
a_{1}&a_{2}&\cdots&a_{n}\\
\vdots&\vdots&\ddots&\vdots\\
a_{1}^{2k-2}&a_{2}^{2k-2}&\cdots&a_{n}^{2k-2}\\
\end{array}
\right)$ and
$X=\left(
\begin{array}{c}
v_{1}^{2}\\
v_{2}^{2}\\
\vdots\\
v_{n}^{2}
\end{array}
\right)=
\left(
\begin{array}{c}
x_{1}\\
x_{2}\\
\vdots\\
x_{n}\\
\end{array}
\right)$.
\end{center}
Since $\mathrm{rank}(A)=n-2k+1$, the vectors
\begin{equation}\label{AX=0}
\left(\frac{1}{L_{\overrightarrow{a}}(a_1)},\ldots,\frac{1}{L_{\overrightarrow{a}}(a_n)}\right)^T,  \left(\frac{a_1}{L_{\overrightarrow{a}}(a_1)},\ldots,\frac{a_n}{L_{\overrightarrow{a}}(a_n)}\right)^T,\ldots, \left(\frac{a_1^{n-2k}}{L_{\overrightarrow{a}}(a_1)},\ldots,\frac{a_n^{n-2k}}{L_{\overrightarrow{a}}(a_n)}\right)^T
\end{equation}
form a basis of the solution space of $AX=\overrightarrow{0}$. Therefore,
\begin{center}
$v_i^2=\sum\limits_{h=0}^{n-2k}\lambda_h a_i^h\cdot L_{\overrightarrow{a}}(a_i)^{-1}\neq0\,\,\text{for any}\,\,1\leq i\leq n,$
\end{center}
where $\lambda_h\in \mathbb{F}_{q}$ for $0\leq h\leq n-2k$. As a result, $v_i^2=\lambda(a_i)L_{\overrightarrow{a}}(a_i)^{-1}\neq0$ for any $1\leq i\leq n$.
\end{proof}
As a corollary of this result, Corollary 2.4 of [\ref{JX2}] can be deduced directly with even $n$ and $k=\frac{n}{2}$ by choosing $\lambda(x)=\lambda$ to be a constant.

\begin{corollary}([\ref{JX2}], Corollary 2.4)\label{y1}
Let $n$ be an even integer and $k=\frac{n}{2}$. If there exists $\lambda\in\mathbb{F}_{q}^{*}$ such that
$\lambda L_{\overrightarrow{a}}(a_{i})^{-1}\in QR_q$ for all $1\leq i\leq n$, then there exists $\overrightarrow{v}=(v_{1},\ldots,v_{n})$ with
$v_{i}^{2}=\lambda L_{\overrightarrow{a}}(a_{i})^{-1}$ such that the code $\mathbf{GRS}_{k}(\overrightarrow{a},\overrightarrow{v})$ defined
in (1) is an MDS self-dual code of length $n$.
\end{corollary}

\begin{lemma}\label{eGRS}
Assume $1\leq k\leq \lfloor\frac{n+1}{2}\rfloor$. The code $\mathbf{GRS}_k(\overrightarrow{a},\overrightarrow{v},\infty)$ is Euclidean self-orthogonal
if and only if there exists a polynomial $\lambda(x)=-x^{n-2k+1}+\lambda_{n-2k}x^{n-2k}+\cdots+\lambda_0\in \mathbb{F}_{q}[x]$ such that
$v_i^2=\lambda(a_i)L_{\overrightarrow{a}}(a_i)^{-1}\neq0$ for any $1\leq i\leq n$.
\end{lemma}

\begin{proof}
Since $\mathbf{GRS}_k(\overrightarrow{a},\overrightarrow{v},\infty)$ is self-orthogonal, one has
\begin{equation}\label{extended GRS}
\begin{cases}
\sum\limits_{i=1}^n v_i^{2}a_{i}^m=0 ,&0\leq m \leq 2k-3;\\
\sum\limits_{i=1}^n v_i^{2}a_{i}^{2k-2}+1=0.
\end{cases}
\end{equation}
Similarly as in Lemma \ref{GRS}, from the first equation of (\ref{extended GRS}),
\begin{equation*}
v_i^2=\sum\limits_{h=0}^{n-2k+1}\lambda_h a_i^h\cdot L_{\overrightarrow{a}}(a_i)^{-1}\,\,\text{for any}\,\,1\leq i\leq n.
\end{equation*}
Substituting to the last equation of (\ref{extended GRS}),
\begin{equation*}
\lambda_{n-2k+1}\sum\limits_{i=1}^na_i^{n-1}\cdot L_{\overrightarrow{a}}(a_i)^{-1}+1=0,
\end{equation*}
which implies $\lambda_{n-2k+1}=-1$ from Lemma \ref{y3}. Hence there exists a polynomial
\begin{equation}\label{polynomial}
\lambda(x)=-x^{n-2k+1}+\lambda_{n-2k}x^{n-2k}+\cdots+\lambda_0\in \mathbb{F}_{q}[x]
\end{equation}
such that
$v_i^2=\lambda(a_i)L_{\overrightarrow{a}}(a_i)^{-1}\neq0$ for any $1\leq i\leq n$.
\end{proof}

As a corollary of this result, Lemma 2 in [\ref{Yan}] can be deduced directly with odd $n$ and $k=\frac{n+1}{2}$ by setting $\lambda(x)=-1$.

\begin{corollary}([\ref{Yan}], Lemma 2)\label{y2}
Let $n$ be odd and $k=\frac{n+1}{2}$. If $-L_{\overrightarrow{a}}(a_{i})^{-1}\in QR_q$ for all $1\leq i\leq n$,
then there exists $\overrightarrow{v}=(v_{1},\ldots,v_{n})$ with $v_{i}^{2}=- L_{\overrightarrow{a}}(a_{i})^{-1}$ such that
$\mathbf{GRS}_{k}(\overrightarrow{a},\overrightarrow{v},\infty)$ defined in (2) is an MDS self-dual code of length $n+1$.
\end{corollary}
\begin{lemma}([\ref{Yan}], Lemma 3)\label{y2}
Let $m\mid q-1$ be a positive integer and $\alpha\in\mathbb{F}_{q}$ be a primitive $m$-th root of unity.
Then for any $1\leq i\leq m$, we have $$\prod_{1\leq j\leq m, j\neq i}\left(\alpha^{i}-\alpha^{j}\right)=m\alpha^{-i}.$$

\end{lemma}


\section{MDS Euclidean Self-dual Codes}

 \quad\; In this section, we present new MDS Euclidean self-dual codes from GRS codes.
Throughout the remaining of this paper, we fix $g$ as a primitive element of $\mathbb{F}_{q}^*$.

\begin{theorem}\label{thmA1}
Let $q=r^{2}$, where $r$ is an odd prime power. For any $1\leq s\leq \frac{r+1}{2}$ and $1\leq t\leq \frac{r-1}{2}$, assume that $n=s(r-1)+t(r+1)$.
There exists a $q$-ary $[n,\frac{n}{2}]$ MDS Euclidean self-dual code, if $r$ and $s$ satisfy one of the following conditions:

(i). $r\equiv1\,(\mathrm{mod}\,4)$ and $s$ is even.

(ii). $r\equiv3\,(\mathrm{mod}\,4)$ and $s$ is odd.
\end{theorem}

\begin{proof}
(i). Let $\alpha=g^{r+1}$, $\beta=g^{r-1}$ and $\gamma=g^{\frac{r+1}{2}}$. Denote by
$\langle\alpha\rangle=\{1,\alpha,\ldots,\alpha^{r-2}\}$ and $\langle\beta\rangle=\{1,\beta,\ldots,\beta^r\}$ and choose
\begin{equation*}
\begin{aligned}
\overrightarrow{a}=\left(\langle\alpha\rangle,\beta\langle\alpha\rangle,\ldots,\beta^{s-1}\langle\alpha\rangle,
\gamma\langle\beta\rangle,\gamma^3\langle\beta\rangle,\ldots,\gamma^{2(t-1)+1}\langle\beta\rangle\right).
\end{aligned}
\end{equation*}
Since $r\equiv1\,(\mathrm{mod}\,4)$, it follows that $\alpha,\beta\in QR_q$ and $\gamma\not\in QR_q$. Therefore,
$\beta^i\langle\alpha\rangle\bigcap\gamma^{2j+1}\langle\beta\rangle=\emptyset$ for any $0\leq i\leq s-1$ and $0\leq j\leq t-1$.
It is necessary to consider two cases.
\begin{itemize}
  \item  For $0\leq i\leq s-1$ and $0\leq j\leq r-2$,
  \begin{equation}\label{L1}
\begin{aligned}
L_{\overrightarrow{a}}(\beta^i\alpha^j)&=\prod\limits_{h=0,h\neq j}^{r-2}(\beta^i\alpha^j-\beta^i\alpha^h)\cdot\prod\limits_{l=0,l\neq i}^{s-1}\prod\limits_{h=0}^{r-2}
(\beta^i\alpha^j-\beta^l\alpha^h)\cdot\prod\limits_{l=0}^{t-1}\prod\limits_{h=0}^r(\beta^i\alpha^j-\beta^h\gamma^{2l+1})&\\
&=\beta^{i(r-2)}\cdot(r-1)\cdot\alpha^{-j}\cdot\prod\limits_{l=0,l\neq i}^{s-1}(\beta^{-2i}-\beta^{-2l})
\cdot\prod\limits_{l=0}^{t-1}\left(\alpha^{j(r+1)}-\gamma^{(2l+1)(r+1)}\right).&\\
\end{aligned}
\end{equation}
Here, the second equality follows from Lemma \ref{y2}.

Note that $\beta,\alpha,r-1\in QR_q$ and $\prod\limits_{l=0}^{t-1}\left(\alpha^{j(r+1)}-\gamma^{(2l+1)(r+1)}\right)\in\mathbb{F}_r^*\subseteq QR_q$. We
only need to consider $u:=\prod\limits_{l=0,l\neq i}^{s-1}(\beta^{-2i}-\beta^{-2l})$. Then
\begin{equation*}
\begin{aligned}
u^r&=\prod\limits_{l=0,l\neq i}^{s-1}(\beta^{2i}-\beta^{2l})=\prod\limits_{l=0,l\neq i}^{s-1}(\beta^{2(i+l)}\cdot(\beta^{-2l}-\beta^{-2i}))&\\
&=\beta^{2\left((s-1)i+\sum\limits_{l=0,l\neq i}^{s-1}l\right)}\cdot(-1)^{s-1}\cdot\prod\limits_{l=0,l\neq i}^{s-1}(\beta^{-2i}-\beta^{-2l})&\\
&=(-1)^{s-1}\cdot\beta^{2\left((s-2)i+\frac{s(s-1)}{2}\right)}\cdot u.&
\end{aligned}
\end{equation*}
So $$u^{r-1}=(-1)^{s-1}\cdot\beta^{2\left((s-2)i+\frac{s(s-1)}{2}\right)}=g^{\frac{(r+1)(s-1)}{2}\cdot(r-1)+2\left((s-2)i+\frac{s(s-1)}{2}\right) \cdot(r-1)},$$
that is
\begin{equation}\label{u1}
\begin{aligned}
u=g^{\frac{(r+1)(s-1)}{2}+2\left((s-2)i+\frac{s(s-1)}{2}\right)+m\cdot(r+1)}
\end{aligned}
\end{equation}
 with some integer $m$.
  \item For $0\leq i\leq t-1$ and $0\leq j \leq r$,
\begin{equation}\label{L2}
\begin{aligned}
L_{\overrightarrow{a}}(\gamma^{2i+1}\beta^j)=&\prod\limits_{h=0,h\neq j}^{r}(\gamma^{2i+1}\beta^j-\gamma^{2i+1}\beta^h)\cdot
\prod\limits_{l=0,l\neq i}^{t-1}\prod\limits_{h=0}^{r}(\gamma^{2i+1}\beta^j-\gamma^{2l+1}\beta^h)&\\
&\cdot\prod\limits_{h=0}^{s-1}\prod\limits_{l=0}^{r-2}(\gamma^{2i+1}\beta^j-\alpha^l\beta^h)&\\
=&\left(\gamma^{r(2i+1)}\cdot(r+1)\cdot\beta^{-j}\right)\cdot\prod\limits_{l=0,l\neq i}^{t-1}\left(\gamma^{(2i+1)(r+1)}-\gamma^{(2l+1)(r+1)}\right)&\\
&\cdot\left((-1)^s\cdot\prod\limits_{h=0}^{s-1}(\beta^{-2j}+\beta^{-2h})\right).&\\
\end{aligned}
\end{equation}
Consider $u:=\prod\limits_{h=0}^{s-1}(\beta^{-2j}+\beta^{-2h})$. Then
\begin{equation*}
\begin{aligned}
u^r&=\prod\limits_{h=0}^{s-1}(\beta^{2j}+\beta^{2h})=\prod\limits_{h=0}^{s-1}\beta^{2(j+h)}(\beta^{-2j}+\beta^{-2h})&\\
&=\beta^{2\left(sj+\sum\limits_{h=0}^{s-1}h\right)}\cdot\prod\limits_{h=0}^{s-1}(\beta^{-2j}+\beta^{-2h})&\\
&=\beta^{2\left(sj+\frac{s(s-1)}{2}\right)}\cdot u.&
\end{aligned}
\end{equation*}
Hence $$u^{r-1}=g^{(r-1)\cdot2\left(sj+\frac{s(s-1)}{2}\right)},$$ that is, $u=g^{2\left(sj+\frac{s(s-1)}{2}\right)+m(r+1)}$  for some integer $m$.
\end{itemize}

Choose $\lambda=\gamma=g^{\frac{r+1}{2}}$.
\begin{itemize}
\item In (\ref{L1}), all of $\beta,\alpha,r-1$ and $\prod\limits_{l=0}^{t-1}\left(\alpha^{j(r+1)}-\gamma^{(2l+1)(r+1)}\right)$ are
squares in $\F_q^*$. From (\ref{u1}) and $s$ is even,
\begin{center}
$\lambda u^{-1}=g^{-\left(\frac{(r+1)(s-2)}{2}+2\left((s-2)i+\frac{s(s-1)}{2}\right)+m\cdot(r+1)\right)}$
\end{center}
is a nonzero square. Thus $\lambda L_{\overrightarrow{a}}(\beta^i\alpha^j)^{-1}\in QR_q$.
\item In (\ref{L2}), it is clear that all elements except $\gamma^{r(2i+1)}$ are nonzero squares. So
$\lambda L_{\overrightarrow{a}}(\gamma^{2i+1}\beta^j)^{-1}\in QR_q$.
\end{itemize}
In summary, by Corollary \ref{y1}, there exists a $q$-ary $\left[n,\frac{n}{2}\right]$ MDS Euclidean self-dual code over $\mathbb{F}_q$.

(ii). Since the proof of case (ii) is similar to the proof of case (i), we omit some details. Let $\alpha=g^{r+1}$, $\beta=g^{r-1}$ and
$\gamma=g^{\frac{r-1}{2}}$. Choose
\begin{equation*}
\begin{aligned}
\overrightarrow{a}=\left(\langle\beta\rangle,\alpha\langle\beta\rangle,\ldots,\alpha^{t-1}\langle\beta\rangle,
\gamma\langle\alpha\rangle,\gamma^{3}\langle\alpha\rangle,\ldots,\gamma^{2(s-1)+1}\langle\alpha\rangle\right).
\end{aligned}
\end{equation*}
Since $r\equiv3\,(\mathrm{mod}\,4)$, it follows that $\alpha,\beta\in QR_q$ and $\gamma\not\in QR_q$. Therefore,
$\alpha^i\langle\beta\rangle\bigcap\gamma^{2j+1}\langle\alpha\rangle=\emptyset$ for any $0\leq i\leq t-1$ and $0\leq j\leq s-1$.
We also consider two cases.
\begin{itemize}
\item For $0\leq i\leq t-1$ and $0\leq j\leq r$,
\begin{equation}\label{L3}
\begin{aligned}
L_{\overrightarrow{a}}(\alpha^i\beta^j)&=\prod\limits_{h=0,h\neq j}^{r}(\alpha^i\beta^j-\alpha^i\beta^h)\cdot\prod\limits_{l=0,l\neq i}^{t-1}\prod\limits_{h=0}^{r}
(\alpha^i\beta^j-\alpha^l\beta^h)\cdot\prod\limits_{h=0}^{s-1}\prod\limits_{l=0}^{r-2}(\alpha^i\beta^j-\alpha^l\gamma^{2h+1})&\\
&=\alpha^{ir}\cdot(r+1)\cdot\beta^{-j}\cdot\prod\limits_{l=0,l\neq i}^{t-1}(\alpha^{i(r+1)}-\alpha^{l(r+1)})
\cdot\prod\limits_{h=0}^{s-1}\left(\beta^{j(r-1)}-\gamma^{(2h+1)(r-1)}\right).&\\
\end{aligned}
\end{equation}
Since $\beta,\alpha,r+1\in QR_q$ and $\prod\limits_{l=0,l\neq i}^{t-1}(\alpha^{i(r+1)}-\alpha^{l(r+1)})\in\mathbb{F}_r^*\subseteq QR_q$, we
only need to consider $u:=\prod\limits_{h=0}^{s-1}\left(\beta^{j(r-1)}-\gamma^{(2h+1)(r-1)}\right)$. Let $\xi=\gamma^{r-1}$. Then
$u=\prod\limits_{h=0}^{s-1}(\xi^{2j}-\xi^{2h+1})$. Therefore,
\begin{equation*}
\begin{aligned}
u^r
=(-1)^{s}\cdot\xi^{-2sj-s^2}\cdot u.
\end{aligned}
\end{equation*}
It implies $$u^{r-1}=(-1)^{s}\cdot\xi^{-2sj-s^2}=g^{\frac{(r+1)s}{2}\cdot(r-1)+(-2sj-s^2)\cdot\frac{(r-1)^2}{2}},$$
that is,
\begin{equation}\label{u2}
\begin{aligned}
u=g^{\frac{(r+1)s}{2}-(s^2+2sj)\cdot\frac{r-1}{2}+m\cdot(r+1)}
\end{aligned}
\end{equation}
for some integer $m$.

\item For $0\leq i\leq s-1$ and $0\leq j\leq r-2$,
\begin{equation}\label{L4}
\begin{aligned}
L_{\overrightarrow{a}}(\gamma^{2i+1}\alpha^j)=&\prod\limits_{h=0,h\neq j}^{r-2}(\gamma^{2i+1}\alpha^j-\gamma^{2i+1}\alpha^h)\cdot
\prod\limits_{l=0,l\neq i}^{s-1}\prod\limits_{h=0}^{r-2}(\gamma^{2i+1}\alpha^j-\gamma^{2l+1}\alpha^h)&\\
&\cdot\prod\limits_{h=0}^{t-1}\prod\limits_{l=0}^{r}(\gamma^{2i+1}\alpha^j-\beta^{l}\alpha^h)&\\
=&\left(\gamma^{(r-2)(2i+1)}\cdot(r-1)\cdot\alpha^{-j}\right)\cdot\prod\limits_{l=0,l\neq
i}^{s-1}\left(\gamma^{(2i+1)(r-1)}-\gamma^{(2l+1)(r-1)}\right)&\\
&\cdot\prod\limits_{h=0}^{t-1}\left(-\alpha^{j(r+1)}-\alpha^{h(r+1)}\right).&\\
\end{aligned}
\end{equation}
\end{itemize}
Let $u=\prod\limits_{l=0,l\neq i}^{s-1}\left(\gamma^{(2i+1)(r-1)}-\gamma^{(2l+1)(r-1)}\right)$ and $\xi=\gamma^{r-1}=g^{\frac{(r-1)^2}{2}}$. Then
$u=\prod\limits_{l=0,l\neq i}^{s-1}(\xi^{2i+1}-\xi^{2l+1}).$ Therefore,
\begin{equation*}
\begin{aligned}
u^r
=(-1)^{s-1}\cdot\xi^{-(s-2)(2i+1)-s^2}\cdot u.
\end{aligned}
\end{equation*}
It follows that $$u^{r-1}=g^{(r-1)\cdot\left(\frac{(r+1)(s-1)}{2}-\frac{r-1}{2}\cdot((s-2)(2i+1)+s^2)\right)},$$ that is,
\begin{center}
$u=g^{\frac{(r+1)(s-1)}{2}-\frac{r-1}{2}\cdot((s-2)(2i+1)+s^2)+m(r+1)}$
\end{center}
with some integer $m$. Since $r\equiv3\,(\mathrm{mod}\,4)$ and $s$ is odd, it is easy to verify that
\begin{center}
$\frac{(r+1)(s-1)}{2}-\frac{r-1}{2}\cdot((s-2)(2i+1)+s^2)+m(r+1)$
\end{center}
is even, which yields $u\in QR_q$.

Choose $\lambda=\gamma=g^{\frac{r-1}{2}}$.
\begin{itemize}
\item In (\ref{L3}), all of $\beta,\alpha,r+1$ and $\prod\limits_{l=0,l\neq i}^{t-1}(\alpha^{i(r+1)}-\alpha^{l(r+1)})$ are nonzero
squares in $\F_q$. From (\ref{u2}), $r\equiv3\,(\mathrm{mod}\,4)$ and $s$ is odd,
\begin{center}
$\lambda u^{-1}=g^{-\left(\frac{(r+1)s}{2}-(s^2+1+2sj)\cdot\frac{r-1}{2}+m\cdot(r+1)\right)}$
\end{center}
is a nonzero square, which yields $\lambda L_{\overrightarrow{a}}(\alpha^i\beta^j)^{-1}\in QR_q$.
\item In (\ref{L4}), from the above analysis, it is clear that all factors except $\gamma^{(r-2)(2i+1)}$ are nonzero squares. As a consequence,
$\lambda L_{\overrightarrow{a}}(\gamma^{2i+1}\alpha^j)^{-1}\in QR_q$.
\end{itemize}
In short, by Corollary \ref{y1}, there exists a $q$-ary $\left[n,\frac{n}{2}\right]$ MDS Euclidean self-dual code over $\mathbb{F}_q$.
\end{proof}

\begin{remark}
In (i)(resp. (ii)) of Theorem \ref{thmA1}, the number of $s$ is $\frac{r-1}{4}$(resp. $\frac{r+1}{4}$) and the number of $t$ is $\frac{r-1}{2}$.
Thus the number of MDS Euclidean self-dual codes is $\frac{r-1}{4}\times\frac{r-1}{2}=\frac{1}{8}\cdot (r-1)^2\approx\frac{1}{8}q$
(resp. $\frac{r+1}{4}\times\frac{r-1}{2}=\frac{1}{8}\cdot (q-1)\approx\frac{1}{8}q$).
\end{remark}

\begin{example}
By utilizing all the previous results in Table 1, $1320$ many MDS Euclidean self-dual codes over $\mathbb{F}_{149^{2}}$ of different lengths can be
constructed. However, by Theorem \ref{thmA1}, there are $\frac{1}{8}\cdot(\sqrt{q}-1)^2=2738$ MDS Euclidean self-dual codes of different lengths.
\end{example}

\begin{example}
For $q=151^{2}$, by Theorem \ref{thmA1}, much more MDS Euclidean self-dual codes can be constructed than all previous works.
The comparison can be presented in the following table.
\begin{center}
\begin{longtable}{|c|c|}  
\caption{Comparison of MDS Euclidean self-dual codes of length $n$ over $\mathbb{F}_{151^2}$} \\ \hline
Source & the number of different lengths $n$ \\  \hline
All previous constructions &  $1500$    \\ \hline
Theorem \ref{thmA1} & $2850$\\ \hline
\end{longtable}
 \end{center}
\end{example}

\begin{remark}
For all the previous results, the number of possible lengths in most cases is bounded by $r$ or some divisor of $r-1$, $r+1$ or $q-1$. Therefore,
it is expected that when $q$ is large, the number of possible lengths is much less than $c\cdot q$ for some constant $c$. However, Theorem \ref{thmA1} will produce about $q/8$ MDS Euclidean self-dual codes of different  lengths which contribute $\frac{q/8}{q/2}=25\%$ to all possible lengths for MDS Euclidean self-dual codes.
\end{remark}

\section{MDS Euclidean Self-orthogonal Codes}

 \quad\; In this section, we apply the criterion(Lemma \ref{GRS}) to construct new MDS Euclidean self-orthogonal codes.

\begin{theorem}\label{thmB1}
Let $q=r^{2}$, where $r$ is an odd prime power. For any $1\leq s\leq \frac{r+1}{2}$ and $1\leq t\leq \frac{r-1}{2}$, assume $n=s(r-1)+t(r+1)$ and
$1\leq k\leq\frac{n}{2}-1$. Then there exists a $q$-ary $[n,k]$ MDS Euclidean self-orthogonal code.
\end{theorem}
\begin{proof}
$\bf{Case\; 1}$: $r\equiv1\,(\mathrm{mod}\,4)$. Recall $\alpha$, $\beta$ and $\gamma$ in the proof of Theorem \ref{thmA1}(i). We also choose
\begin{equation*}
\begin{aligned}
\overrightarrow{a}=\left(\langle\alpha\rangle,\beta\langle\alpha\rangle,\ldots,\beta^{s-1}\langle\alpha\rangle,
\gamma\langle\beta\rangle,\gamma^3\langle\beta\rangle,\ldots,\gamma^{2(t-1)+1}\langle\beta\rangle\right).
\end{aligned}
\end{equation*}
By the proof of Theorem \ref{thmA1}(i),
\begin{equation*}
\begin{aligned}
L_{\overrightarrow{a}}(\beta^i\alpha^j)&=\beta^{i(r-2)}\cdot(r-1)\cdot\alpha^{-j}\cdot g^{\frac{(r+1)(s-1)}{2}+2\left((s-2)i+\frac{s(s-1)}{2}\right)+k\cdot(r+1)}
\cdot\prod\limits_{l=0}^{t-1}\left(\alpha^{j(r+1)}-\gamma^{(2l+1)(r+1)}\right)&
\end{aligned}
\end{equation*}
and
\begin{equation*}
\begin{aligned}
L_{\overrightarrow{a}}(\gamma^{2i+1}\beta^j)=&\left(\gamma^{r(2i+1)}\cdot(r+1)\cdot\beta^{-j}\right)\cdot\prod\limits_{l=0,l\neq i}^{t-1}\left(\gamma^{(2i+1)(r+1)}-\gamma^{(2l+1)(r+1)}\right)&\\
&\cdot\left((-1)^s\cdot g^{2\left(sj+\frac{s(s-1)}{2}\right)+m(r+1)}\right).&\\
\end{aligned}
\end{equation*}
\begin{itemize}
\item When $s$ is even, choose $\lambda(x)=\gamma$. Then
$$\lambda(\beta^i\alpha^j)L_{\overrightarrow{a}}(\beta^i\alpha^j)^{-1},\lambda(\gamma^{2i+1}\beta^j)L_{\overrightarrow{a}}(\gamma^{2i+1}\beta^j)^{-1}\in QR_q.$$
\item When $s$ is odd, choose $\lambda(x)=x$. Then
$$\lambda(\beta^i\alpha^j)L_{\overrightarrow{a}}(\beta^i\alpha^j)^{-1},\lambda(\gamma^{2i+1}\beta^j)L_{\overrightarrow{a}}(\gamma^{2i+1}\beta^j)^{-1}\in QR_q.$$
\end{itemize}
By Lemma \ref{GRS}, it follows that there exists a $q$-ary $[n,k]$ MDS Euclidean self-orthogonal code.

$\bf{Case\; 2}$: $r\equiv3\,(\mathrm{mod}\,4)$. Recall $\alpha$, $\beta$ and $\gamma$ in Theorem \ref{thmA1}(ii). We choose
\begin{equation*}
\begin{aligned}
\overrightarrow{a}=\left(\langle\beta\rangle,\alpha\langle\beta\rangle,\ldots,\alpha^{t-1}\langle\beta\rangle,
\gamma\langle\alpha\rangle,\gamma^3\langle\alpha\rangle,\ldots,\gamma^{2(s-1)+1}\langle\alpha\rangle\right)
\end{aligned}
\end{equation*}
Utilizing the proof of Theorem \ref{thmA1}(ii),
\begin{equation*}
\begin{aligned}
L_{\overrightarrow{a}}(\alpha^i\beta^j)=\alpha^{ir}\cdot(r+1)\cdot\beta^{-j}\cdot\prod\limits_{l=0,l\neq i}^{t-1}(\alpha^{i(r+1)}-\alpha^{l(r+1)})
\cdot g^{\frac{(r+1)s}{2}-(s^2+2sj)\cdot\frac{r-1}{2}+m\cdot(r+1)}
\end{aligned}
\end{equation*}
and
\begin{equation*}
\begin{aligned}
L_{\overrightarrow{a}}(\gamma^{2i+1}\alpha^j)=&\left(\gamma^{(r-2)(2i+1)}\cdot(r-1)\cdot\alpha^{-j}\right)\cdot\prod\limits_{h=0}^{t-1}\left(-\alpha^{j(r+1)}-\alpha^{h(r+1)}\right).&\\
&\cdot g^{\frac{(r+1)(s-1)}{2}-\frac{r-1}{2}\cdot((s-2)(2i+1)+s^2)+m(r+1)}.&\\
\end{aligned}
\end{equation*}
\begin{itemize}
\item When $s$ is odd, choose $\lambda(x)=\gamma$. Then
$$\lambda(\beta^i\alpha^j)L_{\overrightarrow{a}}(\beta^i\alpha^j)^{-1},\lambda(\gamma^{2i+1}\beta^j)L_{\overrightarrow{a}}(\gamma^{2i+1}\beta^j)^{-1}\in QR_q.$$
\item When $s$ is even, choose $\lambda(x)=x$. Then
$$\lambda(\beta^i\alpha^j)L_{\overrightarrow{a}}(\beta^i\alpha^j)^{-1},\lambda(\gamma^{2i+1}\beta^j)L_{\overrightarrow{a}}(\gamma^{2i+1}\beta^j)^{-1}\in QR_q.$$
\end{itemize}
By Lemma \ref{GRS}, it follows that there exists a $q$-ary $[n,k]$ MDS Euclidean self-orthogonal code.
\end{proof}
\begin{remark}
In Theorem \ref{thmB1}, the conditions that $s$ is even in Theorem \ref{thmA1}(i) and $s$ is odd in Theorem \ref{thmA1}(ii) can be removed.
So we can construct approximately $\frac{1}{4}\cdot q$ MDS Euclidean self-orthogonal codes $C$ with length $n$ and dimension $\frac{n}{2}-1$.
which means $C\subseteq C^\perp$ and $\dim(C^\perp)=\dim(C)+2$.
\end{remark}

For odd length, Euclidean almost self-dual codes can be constructed via Lemma \ref{GRS}.

\begin{theorem}\label{thmC1}
Let $q=r^{2}$, where $r$ is an odd prime power. For any $1\leq s\leq \frac{r+1}{2}$ and $1\leq t\leq \frac{r-1}{2}$, assume that $n=s(r-1)+t(r+1)+1$.
Then there exists a $q$-ary $[n,\frac{n-1}{2}]$ MDS Euclidean almost self-dual code, provided that $r$ and $s$ satisfy one of the following conditions:

(i). $r\equiv1\,(\mathrm{mod}\,4)$ and $s$ is odd.

(ii). $r\equiv3\,(\mathrm{mod}\,4)$ and $s$ is even.
\end{theorem}
\begin{proof}
(i). Recall $\alpha$, $\beta$, $\gamma$ and $g$ in Theorem \ref{thmA1}(i). Choose
\begin{equation*}
\begin{aligned}
\overrightarrow{a}=\left(\langle\alpha\rangle,\beta\langle\alpha\rangle,\ldots,\beta^{s-1}\langle\alpha\rangle,
\gamma\langle\beta\rangle,\gamma^3\langle\beta\rangle,\ldots,\gamma^{2(t-1)+1}\langle\beta\rangle,0\right).
\end{aligned}
\end{equation*}
Similarly as the proof of Theorem \ref{thmA1}(i),
\begin{equation*}
\begin{aligned}
L_{\overrightarrow{a}}(\beta^i\alpha^j)&=\beta^{i(r-1)}\cdot(r-1)\cdot g^{\frac{(r+1)(s-1)}{2}+2\left((s-2)i+\frac{s(s-1)}{2}\right)+m\cdot(r+1)}
\cdot\prod\limits_{l=0}^{t-1}\left(\alpha^{j(r+1)}-\gamma^{(2l+1)(r+1)}\right),&
\end{aligned}
\end{equation*}
and
\begin{equation*}
\begin{aligned}
L_{\overrightarrow{a}}(\gamma^{2i+1}\beta^j)=&\left(\gamma^{(r+1)(2i+1)}\cdot(r+1)\right)\cdot\prod\limits_{l=0,l\neq i}^{t-1}\left(\gamma^{(2i+1)(r+1)}-\gamma^{(2l+1)(r+1)}\right)&\\
&\cdot\left((-1)^s\cdot g^{2\left(sj+\frac{s(s-1)}{2}\right)+m(r+1)}\right)&
\end{aligned}
\end{equation*}
A routine calculation shows that
\begin{equation*}
\begin{aligned}
L_{\overrightarrow{a}}(0)&=\prod\limits_{l=0}^{s-1}\prod\limits_{h=0}^{r-2}(0-\beta^l\alpha^h)\cdot\prod\limits_{l=0}^{t-1}\prod\limits_{h=0}^r(0-\beta^h\gamma^{2l+1})&\\
&=\alpha^{\frac{(r-2)(r-1)s}{2}}\cdot\beta^{\frac{(r-1)(s-1)s+r(r+1)t}{2}}\cdot\gamma^{(r+1)t^2}.&
\end{aligned}
\end{equation*}
When $s$ is odd, it is easy to see
$$L_{\overrightarrow{a}}(\beta^i\alpha^j)^{-1},L_{\overrightarrow{a}}(\gamma^{2i+1}\beta^j)^{-1},L_{\overrightarrow{a}}(0)^{-1}\in QR_q.$$
By Lemma \ref{GRS}, it follows that there exists a $q$-ary $[n,\frac{n-1}{2}]$ MDS Euclidean almost self-dual code.

(ii). Recall $\alpha$, $\beta$, $\gamma$ and $g$ in Theorem \ref{thmA1}(ii). Choose
\begin{equation*}
\begin{aligned}
\overrightarrow{a}=\left(\langle\beta\rangle,\alpha\langle\beta\rangle,\ldots,\alpha^{t-1}\langle\beta\rangle,
\gamma\langle\alpha\rangle,\gamma^3\langle\alpha\rangle,\ldots,\gamma^{2(s-1)+1}\langle\alpha\rangle,0\right)
\end{aligned}
\end{equation*}
Similarly as the proof of Theorem \ref{thmA1}(ii),
\begin{equation*}
\begin{aligned}
L_{\overrightarrow{a}}(\alpha^i\beta^j)=\alpha^{i(r+1)}\cdot(r+1)\cdot\prod\limits_{l=0,l\neq i}^{t-1}(\alpha^{i(r+1)}-\alpha^{l(r+1)})
\cdot g^{\frac{(r+1)s}{2}-(s^2+2sj)\cdot\frac{r-1}{2}+m\cdot(r+1)}
\end{aligned}
\end{equation*}
and
\begin{equation*}
\begin{aligned}
L_{\overrightarrow{a}}(\gamma^{2i+1}\alpha^j)=&\left(\gamma^{(r-1)(2i+1)}\cdot(r-1)\right)\cdot\prod\limits_{l=0,l\neq
i}^{s-1}\left(\gamma^{(2i+1)(r-1)}-\gamma^{(2l+1)(r-1)}\right)&\\
&\cdot g^{\frac{(r+1)(s-1)}{2}-\frac{r-1}{2}\cdot((s-2)(2i+1)+s^2)+m(r+1)}.&\\
\end{aligned}
\end{equation*}
Meanwhile,
\begin{equation*}
\begin{aligned}
L_{\overrightarrow{a}}(0)&=\prod\limits_{l=0}^{t-1}\prod\limits_{h=0}^{r}(0-\alpha^l\beta^h)\cdot\prod\limits_{l=0}^{s-1}\prod\limits_{h=0}^{r-2}(0-\gamma^{2l+1}\alpha^h)&\\
&=\alpha^{\frac{(r+1)(t-1)t+(r-2)(r-1)s}{2}}\cdot\beta^{\frac{r(r+1)t}{2}}\cdot\gamma^{(r-1)s^2}.&
\end{aligned}
\end{equation*}
When $s$ is even, it is obvious that
$$L_{\overrightarrow{a}}(\alpha^i\beta^j)^{-1},L_{\overrightarrow{a}}(\gamma^{2i+1}\alpha^j)^{-1},L_{\overrightarrow{a}}(0)^{-1}\in QR_q.$$
By Lemma \ref{GRS}, it follows that there exists a $q$-ary $[n,\frac{n-1}{2}]$ MDS Euclidean  almost self-dual code.
\end{proof}

\begin{remark}
In Theorem \ref{thmC1}, the number of $s$ is about $\frac{r}{4}$ and the number of $t$ is $\frac{r-1}{2}$. Therefore, there exist about
$\frac{1}{8}\cdot q$ MDS Euclidean almost self-dual codes with different lengths.
\end{remark}

\section{Conclusion}

\quad\; We propose a criterion of Euclidean self-orthogonal GRS codes. Based on the criterion, we construct new MDS Euclidean (almost) self-dual codes and
new MDS Euclidean self-orthogonal codes by utilizing GRS codes. For any large square $q$, about $\frac{1}{8}\cdot q$ new MDS Euclidean (almost) self-dual
codes with different lengths can be produced. Our results contribute about $25\%$ possible lengths. But for nonsquare odd prime power $q$, the possible lengths for MDS Euclidean self-dual codes are very restricted and there are still a lot of cases to be explored.

\end{document}